\documentclass[11pt]{article}
\usepackage{amsthm,amsmath,amssymb}
\usepackage[english]{babel}
\usepackage{epsfig}

\newcommand{\R}{{\mathbb R}}

\newcommand{\EE}{{\mathbb E}}
\newcommand{\PP}{{\mathbb P}}
\newcommand{\sN}{{\cal{N}}}
\newtheorem{defn}{Definition}
\newtheorem{thm}{Theorem}
\newtheorem{res}{Lemma}
\newtheorem{cor}{Corollary}

\renewenvironment{proof}{\noindent{\bf Proof:} }{\hfill $\square$ \\}

\begin{document}
\thispagestyle{empty}
\begin{center}
{\Large \sc
Moment Analysis of the Delaunay Tessellation Field Estimator}\\[.5in]

\noindent
{\large M.N.M. van Lieshout}\\[.2in]
\noindent
{\em CWI\/} \& {\em Eindhoven University of Technology\/} \\
\mbox{} \\
P.O. Box 94079, 1090 GB Amsterdam, The Netherlands \\[.2in]
\end{center}

\begin{verse}
{\footnotesize
\noindent
{\bf Abstract}\\
\noindent
The Campbell--Mecke theorem is used to derive explicit expressions for the 
mean and variance of Schaap and Van de Weygaert's Delaunay tessellation field 
estimator. Special attention is paid to Poisson processes.\\[0.2in]

\noindent
{\em Keywords \& Phrases:}
Campbell--Mecke formula, Delaunay tessellation field estimator, generalised
weight function estimator, intensity function, mass preservation, Poisson 
point process, second order factorial moment measure, second order product 
density.

\noindent
{\em 2000 Mathematics Subject Classification:}
60G55, 62M30.
}
\end{verse}
\thispagestyle{empty}

\section{Preliminaries and notation}
\label{S:lambda}

Let $\varphi$ be a locally finite point pattern in $\R^d$ arising as a realisation 
of simple point processes $\Phi$ on $\R^d$ \cite{DaleVere88,Lies00}. In practice, 
$d \in \{1, 2, 3\}$. We shall assume that the points are in general quadratic position 
\cite{Moll94}, that is, 
(a) no $d+2$ points are located on the boundary of a sphere, and 
(b) in the plane no three points are co-linear; in higher dimensions, no $k+1$ 
    points lie in a $k-1$ dimensional affine subspace for $k=2, \dots d$.
These assumptions are satisfied almost surely for realisations of a Poisson process 
with locally finite intensity function $\lambda: \R^d \to [0,\infty)$ or, more 
generally, for Gibbs point processes defined by their probability density with
respect to such a Poisson process. 

Any point pattern $\varphi$ gives rise to two interesting tessellations. First 
consider the set
\[
C(x_i \mid \varphi) := 
   \{ y \in \R^d : || x_i - y || \leq || x_j - y || \quad \forall x_j \in \varphi \}
\]
that consists of all points in $\R^d$ that are at least as close to $x_i \in \varphi$ 
as to any other point of $\varphi$, which is called the {\em Voronoi cell\/} of $x_i$. 
The ensemble of all Voronoi cells is the {\em Voronoi tessellation\/} of $\varphi$ 
\cite{Voro08}. An equivalent definition is
\[
C(x_i \mid \varphi ) = \bigcap_{x_j\neq x_i\in \varphi} H(x_i, x_j),
\]
where $H(x_i, x_j)$ is the closed halfspace 
$\{ y \in \R^d : \left< y - ( x_i + x_j ) / 2, x_i - x_j \right> \geq 0 \}$ 
consisting of points that are at least as close to $x_i$ as to $x_j$. 
In $\R^1$, for $x_i < x_j$, $H(x_i, x_j) = ( - \infty, (x_i + x_j) / 2 ]$. In 
the plane, $H(x_i, x_j)$ is the closed halfplane bounded by the bisecting line
$L(x_i, x_j)$ of the segment connecting $x_i$ and $x_j$ that contains $x_i$.
Note that the Voronoi cells are closed and convex, but not necessarily
bounded.

Under our assumptions, intersections between $k = 2, \dots, d+1$ different
Voronoi cells are either empty or of dimension $d - k + 1$. In particular,
\[
\bigcap_{i=1}^{d+1} C(x_i  \mid \varphi) \neq \emptyset 
\Leftrightarrow b(x_1, \dots, x_{d+1}) \cap \varphi = \emptyset
\]
where $b(x_1, \dots, x_{d+1})$ is the open ball spanned by 
$x_1, \dots x_{d+1}$ on its boundary, and in that case is a single
point, usually referred to as a {\em vertex\/} of the Voronoi diagram. 

Vertices can be used to define the second tessellation of interest to
us in this paper, the {\em Delaunay tessellation\/}. Indeed, suppose that 
$\varphi$ contains at least $d+1$ points. Each Voronoi vertex arising as 
the intersection of $d+1$ cells $C(x_i \mid \varphi)$ defines a closed 
simplex, the convex hull of $\{ x_1, \dots, x_{d+1}\}$, which is called a 
{\em Delaunay cell\/} \cite{Dela34} and denoted by $D(x_1, \dots, x_{d+1})$. 
Note that for $d=1$, Delaunay cells are intervals, whilst in the plane they 
form triangles.
An alternative, equivalent, edge based construction is to join points 
$x_1, x_2 \in \varphi$ that share a common Voronoi border 
$C(x_1 \mid \varphi) \cap C(x_2 \mid \varphi) \neq \emptyset$ into a
Delaunay edge. In this case, $x_1$ and $x_2$ are called {\em Voronoi
neighbours\/}. The set of neighbours of $x_1$ in $\varphi$ is denoted by
$\sN(x_1 \mid \varphi )$. Either way, the partition of space formed by the 
Delaunay cells is referred to as the {\em Delaunay tessellation\/}. The 
union of Delaunay cells containing $x_i \in \varphi$ is known as the {\em 
contiguous Voronoi cell\/} $W( x_i \mid \varphi )$ of $x_i$ in $\varphi$.

\begin{figure}[hbt]
\begin{center}
\epsfxsize=0.50\hsize
\epsffile{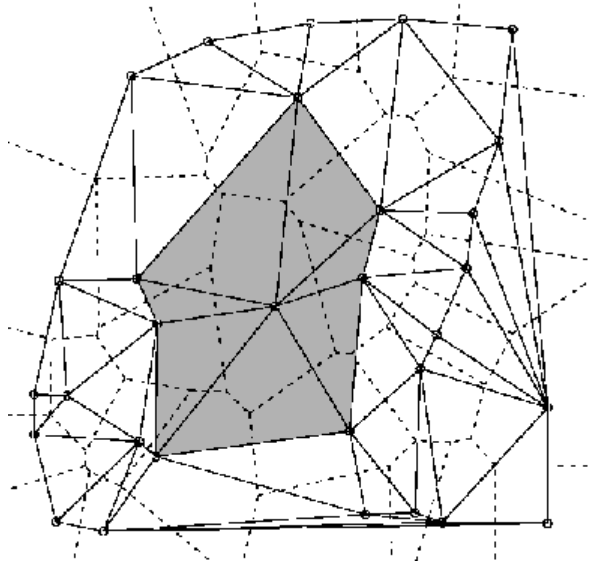}
\end{center}
\caption{A set of thirty points with their Voronoi (dashed lines) and
Delaunay (solid lines) tessellations. A contiguous Voronoi cell is
indicated by shading.}
\label{F:tessellation}
\end{figure}

For more details, including an historical account, the reader is referred 
to the comprehensive textbooks \cite{Moll94,Okabetal92}. An illustration
is given in Figure~\ref{F:tessellation}, which was obtained using the
DELDIR package \cite{deldir}.

\section{Delaunay tessellation field estimator}
\label{S:prelim}

Recently, Schaap and Van de Weygaert \cite{SchaWeyg00,Scha07}
proposed to estimate the intensity function of a spatial point process 
by the so-called {\em Delaunay tessellation field estimator\/} (DTFE). 
The method estimates the intensity at points in a realisation reciprocal to 
the volume of their contiguous Voronoi cell, and distributes these estimated 
field values over Delaunay cells by linear (or other) interpolation. They 
also consider interpolation of fields $x \mapsto f(x) \in \R^+$ observed 
at sampling points. Earlier suggestions to use Voronoi tessellations for
field interpolation include those by Ord \cite{Ord78} and Sibson \cite{Sibs81}.

Based on extensive simulations, Schaap and Van de Weygaert claim that, in 
contrast to kernel estimators \cite{BermDigg89}, the DTFE preserves the total 
mass of the field and fine structural details, appears to result in smooth 
interpolation, adapts itself to the local scale and geometry, and is relatively 
robust. The limitations of the method lie in its sensitivity to measurement 
error, boundary effects, and triangular artefacts \cite{Scha07}. Our aim in 
this paper is a rigorous analysis of this estimator.

Throughout this paper, let $\Phi$ be a simple point process on $\R^d$ having
realisations in general quadratic position for which the expected number of 
points placed in bounded Borel sets is finite so that its (first order) moment 
measure exists as a $\sigma$-finite Borel measure. Furthermore, assume that the 
moment measure is absolutely continuous with respect to Lebesgue measure with 
Radon--Nikodym derivative $\lambda: \R^d \to [0,\infty)$, its {\em intensity 
function\/}.

\begin{defn}
\label{d:DTFE}
Consider a point process $\Phi$ observed in a convex bounded Borel subset $A$ of 
$\R^d$. For $x\in \Phi\cap A$, define
\begin{equation}
\label{e:DTFE}
\widehat { \lambda(x) } := \frac{d+1}{ | W(x \mid \Phi \cap A) | },
\end{equation}
where $| \cdot |$ denotes $d$-volume. For any $x_0 \in A$ in the interior of some 
Delaunay cell, define 
\begin{equation}
\widehat {\lambda(x_0)} := \frac{1}{d+1} 
\sum_{x \in \Phi \cap D( x_0 \mid \Phi\cap A) } \widehat { \lambda(x) }
\label{e:interpolate}
\end{equation}
as the average of the estimated intensity function values at the $d+1$ vertices 
$x$ of the Delaunay cell $D(x_0 \mid \Phi\cap A)$ containing $x_0$. 
\end{defn}

A few remarks are in order. Should a particular realisation $\varphi$ of $\Phi$ 
happen to contain less than $d+1$ points in $A$, the intensity function estimate 
may be set to zero, or (cf.\ the Lemma below) to the number of points divided by 
$|A|$. On the sides of the Delaunay cells, any averaging may be used -- it is a 
null set. Finally, $\widehat{\lambda(x_0) }$ is set to zero for points that do
not fall in any Delaunay cell.

Edge effects arise due to the fact that $\Phi$ is not observed, only $\Phi\cap A$, 
the Delaunay tessellation of which partitions the convex hull of $\Phi \cap A 
\subseteq A$. Such effects may be dealt with in many ways. For example, one might 
use torus corrections, add arbitrary points on the boundary of $A$ (the corners for 
example in the generic case of a cube), or draw lines orthogonal to the edges emanating 
from points on the boundary of the convex hull, etc. Further examples can be found 
in chapter~6 of \cite{Okabetal92}. 

\begin{res} (Schaap and Van de Weygaert \cite{SchaWeyg00,Scha07}) \\
Let $\varphi$ be a realisation of the simple point process $\Phi$ 
containing at least $d+1$ points in $A$. Then the estimator of
Definition~\ref{d:DTFE} preserves total mass, that is,
\[
\int_A \widehat {\lambda(x_0)} dx_0 = n(\varphi \cap A), 
\]
the number of points of $\varphi$ in $A$.
\label{l:mass}
\end{res}

\begin{proof}
Write ${\cal{D}}(\varphi\cap A)$ for the family of Delaunay
cells defined by $\varphi \cap A$, and note that
\begin{eqnarray*}
\int_A \widehat{\lambda(x_0)} \, dx_0  & = & 
\sum_{D_j \in {\cal{D}}(\phi\cap A) } | D_j | \,
\left[ 
\sum_{x \in \varphi \cap D_j} \frac{1}{ | W( x \mid \varphi \cap A ) | }
\right] \\
& = &
\sum_{x \in \varphi \cap A }  \frac{1}{ | W( x \mid \varphi \cap A ) | }
\left[
\sum_{D_j \in {\cal{D}}(\phi\cap A) } 1\{ x \in D_j \} \, | D_j | \right]
= n(\varphi\cap A),
\end{eqnarray*}
cf.\ \cite[p.~62~ff.]{Scha07}.
\end{proof}

\section{Mean and variance of the Delaunay tessellation field estimator}

In this section, we derive the first two moments of the Delaunay tessellation
field estimator. Our first result concerns the expectation.

\begin{thm}
\label{t:DTFE-mean}
Let $\Phi$ be observed in a convex bounded Borel subset $A$, and, for a
point pattern $\varphi$ with $n(\varphi \cap A) \geq d+1$ in general 
quadratic position, set
\begin{equation}
\label{e:DTFEkernel}
g(x_0 \mid x, \varphi) := \frac{ 
\sum_{D_j \in {\cal{D}}(\phi\cap A) }
1\{ x_0 \in D_j^\circ; x \in D_j \}
 }{ | W(x \mid \varphi \cap A ) | },
\end{equation}
for $x_0 \in A \setminus \varphi$, $x\in \varphi$, and let
$g(x \mid x, \varphi) := (d+1) / |W(x \mid \varphi \cap A) |$ if 
$x\in \phi \cap A$. Then the Delaunay tessellation field
estimator defined by (\ref{e:interpolate}) and 
(\ref{e:DTFE}) has expectation
\[
\EE\left[ \widehat{\lambda(x_0)} \right] =
\int_A \EE_x \left[ g(x_0 \mid x, \Phi ) \right]
 \, \lambda(x) \, dx,
\]
where $\EE_x$ denotes the expectation with respect to the Palm 
distribution of $\Phi$ at $x$.
\end{thm}

For patterns $\varphi$ with less than $d+1$ points falling in $A$, it 
is also possible to write 
$\widehat{\lambda(x_0)} = \sum_{x \in \varphi \cap A} g(x_0 \mid x, \varphi )$
with the function $g$ chosen to suit the particular type of edge correction 
adopted, see Section~\ref{S:prelim}.

\

\begin{proof}
Note that 
\[
\widehat{\lambda(x_0)} = \sum_{x \in \Phi \cap A} g(x_0 \mid x, \Phi ).
\]
Hence, by the Campbell--Mecke theorem \cite{SKM}, 
\[
\EE\widehat { \lambda(x_0) }  = \int_A \EE_x \left[
g(x_0 \mid x, \Phi ) \right] \, \lambda(x) \, dx.
\]
\end{proof}

Recall that the second order {\em factorial moment measure\/} 
$\mu^{(2)}$ is defined in integral terms by
\begin{equation}
\label{e:factorial}
\EE \left[ \sum_{x_1, x_2 \in \Phi}^{\neq} f(x_1, x_2) \right]
= \int \int f(x_1, x_2) \, d\mu^{(2)}(x_1, x_2)
\end{equation}
for any non-negative measurable function $f$. The sum is over all 
pairs of different points. We shall say that the second order factorial
moment measure exists, if it is locally finite. If furthermore $\mu^{(2)}$
is absolutely continuous with respect to the $2$-fold product measure of 
Lebesgue measure with itself, a Radon--Nikodym derivative exists known
as second order {\em product density\/} and denoted by $\rho^{(2)}$.
In this case, (\ref{e:factorial}) reduces to
\[
\int \int f(x_1, x_2) \, \rho^{(2)}(x_1, x_2) \, dx_1 \, dx_2.
\]

\begin{thm}
\label{t:DTFE-var}
Let $\Phi$ be observed in a convex bounded Borel subset $A$ and define the 
function $g$ by (\ref{e:DTFEkernel}). Assume that second order product 
densities exist. Then the Delaunay tessellation field estimator defined by 
(\ref{e:interpolate}) and (\ref{e:DTFE}) has variance
\begin{eqnarray*}
{\rm{Var}} ( \widehat{\lambda(x_0)} ) & = &
\int_A \int_A \EE^{(2)}_{x,y} \left[ 
    g(x_0 \mid x, \Phi) \, g(x_0 \mid y, \Phi ) 
\right] \rho^{(2)}(x,y) \, dx \, dy \\
& + &
\int_A \EE_x \left[ 
   g^2(x_0 \mid x, \Phi ) \right] \lambda(x) \, dx 
 - 
\left( \int_A \EE_x \left[ g(x_0 \mid x, \Phi ) \right] \,
\lambda(x) \, dx \right)^2,
\end{eqnarray*}
where $E^{(2)}_{x,y}$ denotes the two-fold Palm distribution of $\Phi$.
\end{thm}

\begin{proof}
Remark that 
\[
\EE\left[ \widehat{\lambda(x_0)}^2 \right]  = 
\EE\left[ \sum_{x,y\in \Phi\cap A}^{\neq} g(x_0 \mid x, \Phi ) 
\, g(x_0 \mid y, \Phi ) \right]  + 
\EE \left[ \sum_{x\in\Phi\cap A} g^2(x_0 \mid x, \Phi ) \right] .
\]
The cross term on the right hand side is equal to 
\[
\int_A \int_A \EE^{(2)}_{x,y} \left[  
g(x_0 \mid x, \Phi) \, g(x_0 \mid y, \Phi)  \right]  \,
\rho^{(2)}(x,y) \, dx \, dy,
\]
see e.g.\ \cite{DaleVere88},
where $\EE^{(2)}_{x,y}$ denotes the two-fold Palm distribution of
$\Phi$ \cite{Hani82}.  
Another appeal to the Campbell--Mecke theorem yields 
\[
\EE \left[ \sum_{x\in\Phi\cap A} g^2(x_0 \mid x, \Phi ) \right]
= \int_A \EE_x \left[ g^2(x_0 \mid x, \Phi ) \right] \lambda(x) \, dx.
\]
Finally, the variance is obtained using Theorem~\ref{t:DTFE-mean}.
\end{proof}


%

In general, the integrals involved in 
Theorems~\ref{t:DTFE-mean}--\ref{t:DTFE-var} must be evaluated by
numerical or simulation methods. 

\section{Comparison to a classic estimator}

The classic estimator of intensity is the kernel estimator 
\begin{equation}
\label{e:BD}
\widehat{\lambda_{BD}(x_0)} := \frac{n(\Phi \cap b(x_0, h)\cap A)}{
   |b(x_0,h)\cap A|} , \quad x_0 \in A.
\end{equation}
proposed by Berman and Diggle \cite{BermDigg89}. The estimator can be
regarded as a kernel estimator \cite{Silv86} with $k_h(x_0 \mid x ) = 
1\{ || x - x_0 || < h \} / | b(x_0, h) \cap A |$, where $b(x_0, h)$ 
denotes the open ball around $x_0$ with radius $h>0$. The choice of 
bandwidth $h$ determines the amount of smoothing. 

Note that when the bounded observation window $A\neq \emptyset$ is open, 
one never divides by zero. In fact, a stronger statement can be made. The 
function $x \mapsto |b(x,h)\cap A|$ is continuous 
and attains its minimum on the closure $\bar A$. Since any point on the
boundary $\partial A$ always has a neighbour within distance $h$ in $A$,
$\inf_{x\in A} |b(x,h)\cap A| > 0$. Further details may be found e.g.\ in 
\cite{Cres91,Digg83,SKM}.

Although (\ref{e:BD}) is a natural estimator, it does not necessarily
preserve the total mass in $A$ \cite{Scha07}, nor is it based on a 
generalised weight function \cite{Silv86}. It is not hard to modify the 
edge correction in (\ref{e:BD}) to define an estimator \cite{Lies07}
that does preserve total mass and is based on a weight function.

\begin{defn}
\label{d:M}
Consider a point process $\Phi$ observed in an open bounded Borel 
subset $A$ of $\R^d$. For $x_0 \in A$, define
\begin{equation}
\label{e:M}
\widehat{\lambda_K(x_0)} := \sum_{x\in\Phi\cap A} \frac{
1\{ || x-x_0 || < h \}}{|b(x_,h)\cap A|}.
\end{equation}
\end{defn}

\begin{res} 
The estimator of Definition~\ref{d:M} is a generalised weight function
estimator with kernel $k_h(x_0 \mid x ) =  1\{ || x-x_0 || < h \} /
|b(x_,h)\cap A|$ that preserves total mass, that is,
\[
\int_A \widehat {\lambda_K(x_0)} \, dx_0 = n(\Phi \cap A), 
\]
the number of points of $\Phi$ in $A$.
\end{res}

\begin{proof}
Note that 
\[
\int_A k_h(x_0 \mid x ) \, dx_0 = 
\int_A \frac{1\{ || x-x_0 || < h \}}{|b(x,h)\cap A|} \; dx_0 \equiv 1
\]
for all $x\in A$, that is, $\widehat{\lambda_K(\cdot)}$ is a generalised
weight function estimator. Furthermore, for any realised point pattern 
$\varphi$, the restriction $\varphi \cap A$ in $A$ is finite and
\[
\int_A \left[ \sum_{x\in\varphi\cap A} \frac{
1\{ || x-x_0 || < h \}}{|b(x_,h)\cap A|} \right] dx_0 =
\sum_{x\in\varphi\cap A} 
\int_A \frac{1\{ || x-x_0 || < h \}}{|b(x_,h)\cap A|} \; dx_0
= n( \varphi \cap A ).
\]
\end{proof}

Note that the Delaunay tessellation field estimator is based on an
{\em adaptive kernel\/} (\ref{e:DTFEkernel}) as it depends on the
underlying point pattern. Indeed, for every $x\in A$,
\[
\int_A g(x_0 \mid x, \phi) \; dx_0 = 
\int_A \frac{ 
\sum_{D_j \in {\cal{D}}(\phi\cap A) }
1\{ x_0 \in D_j^\circ; x \in D_j \}
 }{ | W(x \mid \varphi \cap A ) | } \; dx_0 = 1.
\]
A clear advantage is that the problem of choosing the bandwidth is 
avoided.

In order to assess the quality of the estimator, we proceed to compute its
mean and variance. 

\begin{thm}
\label{t:kernel-mean}
Let $\Phi$ be observed in a bounded open Borel subset $A$. Then, the 
estimator of Definition~\ref{d:M} has expectation
\[
\EE\left[ \widehat{\lambda_K(x_0)} \right] =
\int_A \frac{ 1\{ x \in b(x_0,h) \}}{ |b(x,h)\cap A| } \; \lambda(x) \, dx .
\]
\end{thm}

\begin{proof}
By the Campbell--Mecke theorem 
\[
\EE\left[ \widehat{\lambda_K(x_0)} \right] = 
\EE\left[ \sum_{x\in\Phi\cap A} \frac{
1\{ || x-x_0 || < h \}}{|b(x_,h)\cap A|} \right] =
\int_A \frac{ 1\{ || x-x_0 || < h \} }{|b(x,h)\cap A|} \; \lambda(x) \; dx.
\]
\end{proof}

If we compare Theorem~\ref{t:kernel-mean} to Theorem~\ref{t:DTFE-mean},
the Palm expectation $\EE_x \left[ g(x_0 \mid x, \Phi \right]$ is replaced by
$k_h(x_0 \mid x)$, as the latter does not depend on the point process $\Phi$.

\begin{thm}
\label{t:kernel-var}
let $\Phi$ be observed in a bounded open Borel subset $A$ and assume that 
second order product densities exist. Then
\[
{\rm{Var}} ( \widehat{\lambda_K(x_0)} )  = 
\int\int_{(b(x_0,h)\cap A)^2} 
   \frac{\rho^{(2)}(x, y) - \lambda(x) \lambda(y)}
        {|b(x,h)\cap A| \; |b(y,h)\cap A|} \; dx \, dy + 
\int_{b(x_0,h)\cap A} \frac{\lambda(x)}{|b(x,h)\cap A|^2} \; dx.
\]
\end{thm}

\begin{proof}
Regarding the second moment, note that
\begin{eqnarray*}
\EE\left[ \widehat{\lambda_K(x_0)}^2 \right] & = & 
\EE\left\{ \sum_{x,y\in \Phi\cap A}^{\neq} \left[ 
\frac{1\{ || x-x_0 || < h \}}{|b(x,h)\cap A|} \;
\frac{1\{ || y-x_0 || < h \}}{|b(y,h)\cap A|} \right] \right\}  \\
& + &
\EE \left\{ \sum_{x\in\Phi\cap A} 
\left[ \frac{1\{ || x-x_0 || < h \}}{|b(x,h)\cap A|^2} \right] \right\}.
\end{eqnarray*}
Then rewrite the expectations as integrals with respect to $\rho^{(2)}$
and $\lambda$ respectively to obtain that the variance of 
$\widehat{\lambda_K(x_0)}$ is equal to
\[
\int_{b(x_0,h)\cap A} \int_{b(x_0,h)\cap A} \frac{1}
   {|b(x,h)\cap A| \; |b(y,h)\cap A|} \; \rho^{(2)}(x, y) \, dx \, dy 
 + 
\int_{b(x_0,h)\cap A} \frac{\lambda(x)}{|b(x,h)\cap A|^2} \; dx.
\]
An appeal to Theorem~\ref{t:kernel-mean} completes the proof.
\end{proof}

The result should be compared to that of Theorem~\ref{t:DTFE-var}.

Similar arguments as those in the proofs of Theorems~\ref{t:kernel-mean}
and \ref{t:kernel-var} applied to the classic Berman--Diggle estimator
(\ref{e:BD}) give mean
\[
\frac{1}{|b(x_0,h) \cap A|} \int_{b(x_0,h)\cap A} \lambda(x) \, dx
\]
and variance
\[
\frac{1}{|b(x_0,h) \cap A|^2} \left\{
\int_{b(x_0,h)\cap A)} \lambda(x) \, dx +
 \int_{(b(x_0,h)\cap A)^2} 
\left[ \rho^{(2)}(x, y) - \lambda(x) \, \lambda(y) \right] \, dx \, dy
\right\} .
\]
Note that for $x_0\in A\ominus b(0,2h)$ separated by $2h$ from the boundary
of $A$, no edge correction is necessary, and both kernel estimators are identical.


The disadvantage of kernel estimators is that they involve a bandwidth 
parameter $h$; the larger $h$, the smoother the estimated intensity function. 
For specific models, $h$ may be chosen by optimisation of the (integrated) 
mean squared error \cite{Digg83}. In practice, in a planar setting,
Diggle \cite{Digg83} recommends to choose $h$ proportional to $n^{-1/2}$, 
where $n$ is the observed number of points. For a fixed bandwidth, neither 
the Berman--Diggle estimator nor the modification of Definition~\ref{d:M} 
is universally better. For examples, the reader is referred to \cite{Lies07}.

\section{Intensity estimation for Poisson point processes}

In general, the integrals involved in 
Theorems~\ref{t:DTFE-mean}--\ref{t:kernel-var} have to be evaluated
numerically. An exception is the case where $\Phi$ is a Poisson point
process with a locally finite intensity function.

\begin{cor}
\label{c:Poisson}
Let $\Phi$ be a Poisson point process observed in a convex bounded Borel
subset $A$. Then,
\[
\EE\left[ \widehat{\lambda(x_0)} \right] =
\int_A \EE \left[ g(x_0 \mid x, \Phi \cup \{ x \} ) \right]
 \, \lambda(x) \, dx
\]
and
\begin{eqnarray*}
{\rm{Var}} ( \widehat{\lambda(x_0)} ) & = &
\int_A \int_A \EE \left[ g(x_0 \mid x, \Phi \cup \{ x, y \} ) \,
                         g(x_0 \mid y, \Phi \cup \{ x, y \} ) \right] 
\, \lambda(x) \, \lambda(y) \, dx \, dy \\
& + &
\int_A \EE \left[ 
   g^2(x_0 \mid x, \Phi \cup \{ x \}) \right] \lambda(x) \, dx 
 - 
\left( \int_A \EE \left[ g(x_0 \mid x, \Phi \cup \{ x \} ) \right] \,
\lambda(x) \, dx \right)^2.
\end{eqnarray*}
\end{cor}

\begin{proof}
For a Poisson process, the Palm distribution at $x$  is equal to
the superposition of its distribution $\PP$ with an extra point at $x$,
the two-fold Palm distribution $\PP^{(2)}_{x,y}$ is the superposition
of $\PP$ with $x$ and $y$. Furthermore, 
$\rho^{(2)}(x, y) = \lambda(x) \, \lambda(y)$ is a product density. 
Plugging these results into the expressions of 
Theorems~\ref{t:DTFE-mean}--\ref{t:DTFE-var} completes the proof.
\end{proof}

\begin{cor}
let $\Phi$ be a Poisson point process observed in a bounded open Borel 
subset $A$ and assume that second order product densities exist. Then,
\[
{\rm{Var}} ( \widehat{\lambda_K(x_0)} )  = 
\int_{b(x_0,h)\cap A} \frac{\lambda(x)}{|b(x,h)\cap A|^2} \; dx.
\]
\end{cor}

\begin{proof}
Use that $\rho^{(2)}(x, y) = \lambda(x) \, \lambda(y)$ and apply
Theorem~\ref{t:kernel-var}.
\end{proof}

The variance of the Berman--Diggle estimator is 
\(
\int_{b(x_0,h)\cap A} {\lambda(x)} \, dx / {|b(x_0,h)\cap A|^2} .
\)

\medskip

For stationary Poisson processes, even more can be said. In the
remainder of this section, define $g$ as in (\ref{e:DTFEkernel}) with 
$A = \R^d$.

\begin{thm}
\label{t:unbiased}
Let $\Phi$ be a stationary Poisson point process in $\R^d$ with intensity 
$\lambda > 0$. Then, the Delaunay tessellation field estimator  
$\widehat{ \lambda(0) }$ is asymptotically unbiased. 
\end{thm}

\begin{proof} 
Let $b(x, y_1, \dots, y_d)$ be the open ball spanned by the points $x$, 
$y_1, \dots, y_d$ on its topological boundary, and let
$D^\circ(x, y_1, \dots, y_d)$ be the {\em open\/} simplex that is the
interior of the convex hull of $\{ x, y_1, \dots, y_d \}$.  Recall that 
the points $x, y_1, \dots, y_d$ define a Voronoi vertex, or, equivalently, 
a Delaunay cell if and only if there are no points in $b(x, y_1, \dots, y_d)$. 

By Corollary~\ref{c:Poisson}, asymptotically
\begin{eqnarray*}
\EE \left[ \widehat{ \lambda(0) } \right] & = &
\lambda \int_{\R^d} \EE \left[ g(0 \mid x, \Phi \cup \{ x \}) \right] \, dx \\
& = &
\lambda \, \int \EE \left[ \sum^{\neq}_{\{y_1, \dots, y_d\} \subset \Phi}
\frac{ 1\{ 0 \in D^\circ(x, y_1, \dots, y_d); 
           b(x, y_1, \dots, y_d) \cap ( \Phi \cup \{ x \} ) = \emptyset \}
}{ |W( x \mid  \Phi \cup \{ x \}) | } \right]  dx \\
& = &
\lambda \, \int \EE \left[ \sum^{\neq}_{\{y_1, \dots, y_d\} \subset \Phi}
\frac{ 1\{ 0 \in D^\circ(x, y_1, \dots, y_d); 
           b(x, y_1, \dots, y_d) \cap \Phi = \emptyset \}
}{ |W( x \mid  \Phi \cup \{ x \}) | } \right]  dx \\
& = & 
\lambda \, \int \EE \left[ \sum^{\neq}_{\{z_1, \dots, z_d\} \subset \Phi_{-x}}
\frac{ 1\{ -x \in D^\circ(0, z_1, \dots, z_d); 
           b(0, z_1, \dots, z_d) \cap \Phi_{-x} = \emptyset \}
}{ |W( 0 \mid  \Phi_{-x} \cup \{ 0 \}) | } \right]  dx \\
& = & 
\lambda \, \int \EE \left[ \sum^{\neq}_{\{z_1, \dots, z_d\} \subset \Phi}
\frac{ 1\{ -x \in D^\circ(0, z_1, \dots, z_d); 
           b(0, z_1, \dots, z_d) \cap \Phi = \emptyset \}
}{ |W( 0 \mid  \Phi\cup \{ 0 \}) | } \right]  dx
\end{eqnarray*}
by stationarity. Hence, by Fubini's theorem,
\begin{eqnarray*}
\EE\left[ \widehat{ \lambda(0) } \right] & = &  
\lambda \, \EE \left[ 
\frac{ \sum^{\neq}_{\{z_1, \dots, z_d\} \subset \Phi}
 | D^\circ(0, z_1, \dots, z_d) | \, 
   1\{ b(0, z_1, \dots, z_d) \cap \Phi = \emptyset \} 
}{ |W( 0 \mid  \Phi\cup \{ 0 \}) | } 
\right]  \\
& = & \lambda \, \EE \left[ 
\frac{ |W( 0 \mid  \Phi\cup \{ 0 \}) | }{ |W( 0 \mid  \Phi\cup \{ 0 \}) | } 
\right] = \lambda.
\end{eqnarray*}
\end{proof}


The asymptotic variance of the Delaunay tessellation field estimator increases 
quadratically with $\lambda$ with a constant multiplier that depends on the 
dimension. The proof rests on the following two lemmata.

\begin{res}
\label{e:C}
Let $\Phi$ be a stationary Poisson point process in $\R^d$ with intensity 
$\lambda > 0$. Then,
\begin{eqnarray*}
C(\lambda, d) & := &\int \int \EE \left[ g(0 \mid x, \Phi \cup \{ x, y \} ) \,
                         g(0 \mid y, \Phi \cup \{ x, y \} ) \right] 
\, \lambda(x) \, \lambda(y) \, dx \, dy \\
& = & 
\lambda^2 \, \int 
\EE_1 \left[ 
\frac{ 
|W( 0 \mid  \Phi \cup \{ 0, x \}) \cap
 W( x \mid  \Phi \cup \{ 0, x \}) |
}{
|W( 0 \mid  \Phi \cup \{ 0, x \}) | \, 
|W( x \mid  \Phi \cup \{ 0, x \}) | } 
\right]  dx,
\end{eqnarray*}
where $\EE_1$ denotes expectation with respect to a unit intensity Poisson
point process.
\end{res}

By the Nguyen--Zessin formula \cite{Lies00}, alternatively
\begin{eqnarray*}
C(\lambda, d) & = & \lambda \, \EE \left[
\frac{1}{
|W( 0 \mid  \Phi \cup \{ 0 \}) |}
\sum_{y\in \sN(0 \mid \Phi \cup \{ 0 \})}
\frac{
|W( 0 \mid  \Phi \cup \{ 0\}) \cap W( y \mid  \Phi \cup \{ 0 \}) | 
}{
|W( y \mid  \Phi \cup \{ 0\}) |
}
\right] \\
& = & \lambda^2 \, \EE_1 \left[
\frac{1}{
|W( 0 \mid  \Phi \cup \{ 0 \}) |}
\sum_{y\in \sN(0 \mid \Phi \cup \{ 0 \})}
\frac{
|W( 0 \mid  \Phi \cup \{ 0\}) \cap W( y \mid  \Phi \cup \{ 0 \}) | 
}{
|W( y \mid  \Phi \cup \{ 0\}) |
}
\right] .
\end{eqnarray*}

\

\begin{proof}
Write $\Phi_{d-1}$ for sets of $d-1$ distinct points in $\Phi$. Then, 
as $\lambda(x) \equiv \lambda$ is constant, and 
$g(0 \mid x, \Phi \cup \{ x, y \} ) \, g(0 \mid y, \Phi \cup \{ x, y \} )$
vanishes when $x$ and $y$ do not belong to the same Delaunay cell
containing $0$ in its interior,
\[
C(\lambda, d)  = 
\lambda^2 \, \int \int 
\EE \left[ 
\sum_{z \in \Phi_{d-1}} \frac{ 
1\{ 0 \in D^\circ(x, y, z ); b(x, y, z) \cap \Phi = \emptyset \}
}{
|W( x \mid  \Phi \cup \{ x, y \}) | \, |W( y \mid  \Phi \cup \{ x, y \}) | } 
\right]  dx \, dy 
\]
\[
 = 
\lambda^2 \, \int \int 
\EE \left[ 
\sum_{z \in \Phi_{-x;d-1}}
\frac{ 
1\{ -x \in D^\circ(0, y-x, z ); b(0, y-x, z) \cap \Phi_{-x} = \emptyset \}
}{
|W( 0 \mid  \Phi_{-x} \cup \{ 0, y-x \}) | \, 
|W( y -x \mid  \Phi_{-x} \cup \{ 0, y-x \}) | } 
\right]  dx \, dy .
\]
Because of stationarity,
\begin{eqnarray*}
C(\lambda, d) & = &
\lambda^2 \, \int \int 
\EE \left[ 
\sum_{z\in \Phi_{d-1}}
\frac{ 1\{ -x \in D^\circ(0, y-x, z); b(0, y-x, z) \cap \Phi = \emptyset \}
}{
|W( 0 \mid  \Phi \cup \{ 0, y-x \}) | \, 
|W( y -x \mid  \Phi \cup \{ 0, y-x \}) | } 
\right]  dx \, dy \\
& = &
\lambda^2 \, \int \int 
\EE \left[ 
\frac{ \sum_{z\in \Phi_{d-1}}
1\{ -x \in D^\circ(0, y, z); b(0, y, z) \cap \Phi = \emptyset \}
}{
|W( 0 \mid  \Phi \cup \{ 0, y \}) | \, 
|W( y \mid  \Phi \cup \{ 0, y \}) | } 
\right]  dx \, dy.
\end{eqnarray*}
Scaling by $\lambda^{1/d}$ yields that $\lambda^{-2} \, C(\lambda, d)$ is 
equal to
\[
\int\int \EE \left[ 
\frac{ \sum_{z\in \Phi_{d-1}}
1\{ - \lambda^{1/d} x \in D^\circ(0, \lambda^{1/d} y, \lambda^{1/d} z); 
     b(0, \lambda^{1/d} y, \lambda^{1/d} z) \cap \lambda^{1/d} \Phi = \emptyset \}
}{
\lambda^{-1} |W( 0 \mid  \lambda^{1/d} \Phi \cup \{ 0, \lambda^{1/d} y \}) | \, 
\lambda^{-1} |W( \lambda^{1/d} y \mid  \lambda^{1/d} \Phi \cup \{ 0,
\lambda^{1/d} y \}) | } \right] dx \, dy.
\]
Since $\lambda^{1/d} \Phi$ is a unit intensity Poisson point process, we obtain
\[
\lambda^{-2} \, C(\lambda, d) = 
\int\int \EE_1 \left[
\frac{ \sum_{ z\in \Phi_{d-1}}
 1\{ - x \in D^\circ(0, y, z);  b(0, y, z) \cap  \Phi = \emptyset \}
}{
 |W( 0 \mid  \Phi \cup \{ 0, y \}) | \, 
 |W( y \mid  \Phi \cup \{ 0, y \}) | } \right] dx \, dy.
\]
An appeal to Fubini's theorem to integrate out over $x$ completes the proof.
\end{proof}

\begin{res}
\label{e:C'}
Let $\Phi$ be a stationary Poisson point process in $\R^d$ with intensity 
$\lambda > 0$. Then,
\[
C^\prime(\lambda, d) := 
\int \EE \left[ g^2(0 \mid x, \Phi \cup \{ x \}) \right] \lambda(x) \, dx 
=
\lambda^2 \, \EE_1 \left[ \frac{1}{ |W( 0 \mid  \Phi\cup \{ 0 \}) | } 
\right].
\]
where $\EE_1$ denotes expectation with respect to a unit intensity Poisson
point process.
\end{res}

\begin{proof}
Using $\lambda(x) = \lambda$ and argueing as in the proof of 
Theorem~\ref{t:unbiased}, we get
\begin{eqnarray}
\nonumber
C^\prime(\lambda, d) 
& = & \lambda \, \int \EE \left[ 
\left( \sum^{\neq}_{\{z_1, \dots, z_d\} \subset \Phi}
\frac{ 1\{ -x \in D^\circ(0, z_1, \dots, z_d); 
           b(0, z_1, \dots, z_d) \cap \Phi = \emptyset \}
}{ |W( 0 \mid  \Phi\cup \{ 0 \}) | } \right)^2 \right]  dx \\
& = &
\label{e:g2}
\lambda \, \int \EE \left[ \sum^{\neq}_{ \{z_1, \dots, z_d\} \subset \Phi }
\frac{ 1\{ -x \in D^\circ(0, z_1, \dots, z_d); 
           b(0, z_1, \dots, z_d) \cap \Phi = \emptyset \}
}{ |W( 0 \mid  \Phi\cup \{ 0 \}) |^2 } \right]  dx 
\end{eqnarray}
as $-x$ belongs to a single Delaunay interior.
Write $\Phi_d$ for sets of $d$ distinct points in $\Phi$ and
scale each point in (\ref{e:g2}) by $\lambda ^{1/d}$ to obtain 
that $C^\prime(\lambda, d)$ is equal to
\[
\lambda \int \EE \left[ 
\sum_{ z \in \Phi_d}
\frac{ 1\{ - \lambda^{1/d} x \in D^\circ(0, \lambda^{1/d} z) ;
       b(0, \lambda^{1/d} z) \cap \lambda^{1/d} \Phi = \emptyset \}
}{ 
 \lambda^{-2} |W( 0 \mid  \lambda^{1/d} \Phi\cup \{ 0 \}) |^2 } \right]  
 dx
\]
which, since $\lambda^{1/d} \Phi$ is a unit rate Poisson process reduces to
\[
=
\lambda^2 \int \EE_1 \left[ 
 \frac{ \sum_{\{z_1, \dots, z_d\} \subset \Phi}
 1\{ - x \in D^\circ(0, z_1, \dots, z_d); 
  b(0, z_1, \dots, z_d) \cap \Phi = \emptyset \}
}{ |W( 0 \mid  \Phi\cup \{ 0 \}) |^2 } \right]  dx .
\]
The sum of $d$-volumes of Delaunay cells involving $0$ is that of its
contiguous Voronoi cell, and we conclude that
\[
C^\prime(\lambda, d) 
=
 \lambda^2 \, \EE_1 \left[ \frac{1}{ |W( 0 \mid  \Phi\cup \{ 0 \}) | } 
\right] .
\]
\end{proof}

The above results can be summarised as follows.

\begin{thm} \label{t:asymptotic}
Let $\Phi$ be a stationary Poisson point process in $\R^d$ with intensity 
$\lambda > 0$. Then, the Delaunay tessellation field estimator 
$\widehat{ \lambda(0) }$ has asymptotic variance $c_d \lambda^2$ with
\[
c_d = 
\EE_1 \left[ 
\frac{1}{
|W( 0 \mid  \Phi \cup \{ 0 \}) |}
\left\{ 
1 + 
\sum_{y\in \sN(0 \mid \Phi \cup \{ 0 \})}
\frac{
|W( 0 \mid  \Phi \cup \{ 0\}) \cap W( y \mid  \Phi \cup \{ 0 \}) | 
}{
|W( y \mid  \Phi \cup \{ 0\}) |
}
\right\} \right] - 1.
\]
\end{thm}

Note that the classic Berman--Diggle estimator (\ref{e:BD}) is asymptotically 
unbiased with variance $\lambda \, \omega_d^{-1} \, h^{-d}$, where $\omega_d$ 
is the volume of the unit ball in $\R^d$. In words, the Berman--Diggle estimator 
is more efficient whenever the average number of points per test ball exceeds 
$1/c_d$.

\section{Poisson processes on the line}
\label{S:Poisson}

For one-dimensional Poisson processes, the distribution of the contiguous 
Voronoi cell can be calculated explicitly for arbitrary intensity functions. 
For simplicity, assume that $A = [-w, w]$ is an interval of radius $w>0$ 
either side of the origin. 

The following lemma is well-known.

\begin{res}
\label{l:lifetime}
Let $\Phi$ be a Poisson point process on $[-w,w]$ with finite intensity
function $\lambda$ and write $\Lambda(a,b) = \int_a^b \lambda(x) \, dx$
for the moment measure of $(a,b)$ for any $-w \leq a \leq b \leq w$. 
For $x\in (-w,w)$, define the random variables
\begin{eqnarray*}
\Phi^-(x) & := & \max \{ y \in \{ -w \} \cup ( \Phi \cap [-w, x) ) \}; \\
\Phi^+(x) & := & \min \{ y \in \{ w \} \cup ( \Phi \cap (x, w] ) \}.
\end{eqnarray*}
Then, their distribution functions are given by
\[
F^-(t)  =  \exp\left[ - \Lambda(t,x) \right]
\]
for $t \in (-w,x)$, with an atom of mass
$P( \Phi^-(x) = -w ) = \exp\left[ - \Lambda(-w,x) \right]$ 
at $-w$, respectively
\[
F^+(s) =  1 - \exp\left[ - \Lambda(x,s) \right]
\] 
for $s \in (x,w)$ with an atom at $w$ of mass
$P( \Phi^+(x) = w )  =  \exp\left[ - \Lambda(x,w) \right]$.
Moreover, for fixed $x$, $\Phi^+(x)$ and $\Phi^-(x)$ are independent 
random variables. 
\end{res}

\subsection{Expectation of the DTFE for Poisson processes on the line}

Note that on the real line, the contiguous Voronoi cell 
$W(x \mid ( \Phi \cup \{ x \} ) \cap [-w,w] )$ is the interval 
$[ \Phi^-(x) , \Phi^+(x) ]$. Thus, Lemma~\ref{l:lifetime} can be used to 
calculate the moments of the Delaunay tessellation field estimator.  In 
this section, we shall deal with edge effects by placing two ghost points 
at the borders $-w$ and $w$.

\begin{thm}
\label{t:mean1D}
Let $\Phi$ be a Poisson point process observed in $A = [-w,w]$ for some
$w>0$ with locally finite intensity function $\lambda : \R \to [0,\infty)$. 
Then, for $x_0 \in A$,
\begin{eqnarray}
\nonumber
\EE\left[ \widehat{\lambda(x_0)} \right] & = &
\int_{-w}^{x_0} \int_{x_0}^w \frac{\Lambda(t,s) \, \lambda(s) \, 
\lambda(t)}{s-t} \, e^{-\Lambda(t,s)} \,  dt \, ds + 
\frac{\Lambda(-w,w) \, e^{-\Lambda(-w,w)}}{2 w}  \\
\label{e:E1DPois}
& + & \int_{x_0}^w \frac{ \Lambda(-w,s) \, \lambda(s)}{w+s} \, 
e^{ -\Lambda(-w,s) } \, ds 
+ \int_{-w}^{x_0} \frac{\Lambda(t,w) \, \lambda(t)}{w-t} \, 
e^{ -\Lambda(t,w) } \, dt.  
\end{eqnarray}
\end{thm}

\begin{proof}
Fix $x_0 \neq x \in (-w,w)$, and let $\varphi$ be a realisation of $\Phi$,
which we augment by $-w$ and $w$ in order to obtain bounded Delaunay
cells. Since almost surely, $x \not \in \Phi$ and $x_0 \not \in \Phi$, 
assume $x_0, x \not \in \varphi$, and consider $g(x_0 \mid x, \varphi 
\cup \{ x \})$ as defined in (\ref{e:DTFEkernel}). 
Note that $x_0$ belongs to a single Delaunay cell interior. If $x$ is 
no endpoint of this cell, $g(x_0 \mid x, \varphi \cup \{ x \}) = 0$. 
Otherwise, $g(x_0 \mid x, \varphi \cup \{ x \}) = 
1 / ( \varphi^+(x_0) - \varphi^-(x_0) )$, cf.\ Lemma~\ref{l:lifetime}.

First, assume $x < x_0$. By Lemma~\ref{l:lifetime} applied to the 
point $x$, 
\[
\EE \left[ g(x_0 \mid x, \Phi \cup \{ x \} ) \right]  =  
\int_{-w}^x \int_{x_0}^w \frac{ dF^-(t) \, dF^+(s) }{ s - t } 
 =  \frac{e^{-\Lambda(-w,w)}}{2 w} + 
\]
\[
+ \int_{x_0}^w \frac{\lambda(s)}{w+s} \, e^{ -\Lambda(-w,s) } \, ds 
 +  
\int_{-w}^x \frac{\lambda(t)}{w-t} \, e^{ -\Lambda(t,w)} \, dt  
 +  
\int_{-w}^x \int_{x_0}^w \frac{\lambda(s) \, \lambda(t)}{s-t} \,
    e^{ -\Lambda(t,s) } \, dt \, ds.
\]
Similarly, for $x_0 < x $, 
\[
\EE \left[ g(x_0 \mid x, \{ x \} \cup \Phi ) \right]  =  
\int_{-w}^{x_0} \int_x^w \frac{ dF^-(t) \, dF^+(s) }{ s - t } \
 =  \frac{e^{-\Lambda(-w,w)}}{2 w} + 
\]
\[
\int_x^w \frac{\lambda(s)}{w+s} \, e^{ -\Lambda(-w,s) } \, ds 
 +  
\int_{-w}^{x_0} \frac{\lambda(t)}{w-t} \, e^{ -\Lambda(t,w) } \, dt 
 + 
\int_{-w}^{x_0} \int_x^w \frac{\lambda(s) \, \lambda(t)}{s-t} \,
   e^{ -\Lambda(t,s) } \, dt \, ds.
\]
By Theorem~\ref{t:DTFE-mean}, the expectation of the Delaunay tessellation
field estimator is as stated for $x_0 \in (-w, w)$.

It remains to consider $x_0 =-w$ or $w$. In the first case, $\varphi^-(x_0)$
is replaced by $-w$; for $x_0 = w$, $\varphi^+(x_0)$ becomes $w$ in the 
evaluation of $g(x_0 \mid x, \varphi \cup \{ x \})$. 
Thus, for example,
\[
\EE \left[ g(-w \mid x, \Phi \cup \{ x \} ) \right]  =  
e^{-\Lambda(-w,x)} \, \int_x^w \frac{ dF^+(s) }{w+s} =
   \frac{e^{-\Lambda(-w,w)}}{2 w} + 
\int_x^w \frac{\lambda(s)}{w+s} \, e^{ -\Lambda(-w,s) } \, ds ,
\]
with a similar expression for $x_0 = w$. Upon integration, (\ref{e:E1DPois})
is obtained, under the convention that integrals over intervals of
zero length vanish.
\end{proof}


In general, (\ref{e:E1DPois})  must be evaluated numerically.
For the homogeneous Poisson process, analytic evaluation is 
possible. In fact, it can be shown that the estimator is unbiased
even near the borders of the observation interval.

\begin{cor}
\label{t:mean1D-stat}
Let $\Phi$ be a stationary Poisson point process observed in $A = [-w,w]$ 
for some $w>0$ with intensity $\lambda > 0$. Then, the Delaunay tessellation 
field estimator $\widehat{ \lambda(x_0) }$ is unbiased for all $x_0 \in A$.
\end{cor}

\begin{proof}
For a stationary Poisson point process, the double integral in
(\ref{e:E1DPois}) reduces to 
\[
\lambda \left( e^{\lambda x_0} - e^{-\lambda w} \right) \times
\left( e^{-\lambda x_0}- e^{-\lambda w} \right)
\]
and in particular vanishes for $x_0 = -w$ or $w$. The three border correction 
terms are equal to $\lambda e^{-2 \lambda w}$, to $\lambda e^{-\lambda w}
( e^{-\lambda x_0} - e^{-\lambda w} )$, and to $\lambda e^{-\lambda w} 
( e^{\lambda x_0} - e^{-\lambda w} )$, respectively. The sum of all four terms 
is $\lambda$, so the estimator is unbiased.
\end{proof}

Note that the Berman--Diggle estimator is unbiased as well, but that 
this may not be true for (\ref{e:M}) due to edge correction near 
the border.

\subsection{Variance of the DTFE for Poisson processes on the line}

In this section, we derive the asymptotic variance of the Delaunay 
tessellation field estimator for a stationary Poisson process on the 
line. The result can be used to approximate the variance when the
underlying intensity function is smoothly varying.

\begin{thm}
\label{t:var1D-stat}
Let $\Phi$ be a stationary Poisson point process observed in $A = [-w, w]$ 
for some $w>0$ with intensity $\lambda > 0$. Then, as  $w\to \infty$,
the Delaunay tessellation field estimator $\widehat{ \lambda(0) }$ has 
asymptotic variance
\(
 2 \, \lambda^2 (2-\pi^2/6 ) \approx 0.7 \lambda^2 .
\)
\end{thm}

The result should be compared to $\lambda / (2h)$ for the Berman--Diggle
kernel estimator \cite{BermDigg89}, see also \cite{Lies07}. If 
$2 \lambda h > 1.4$, that is the average number of points per bin at least 
$1.4$, kernel estimation is the better choice. Naturally, in order to 
compute $\widehat{ \lambda(x_0) }$, two points of the underlying process 
are used.


In order to give the proof, some special function theory is needed.
Let $x>0$. Recall that the exponential integral is defined as
\[
E_1(x) = \int_1^\infty \frac{e^{-tx}}{t} dt 
       = \int_x^\infty \frac{e^{-u}}{u} du.
\]
Its integral satisfies
\[
E_2(x) = \int_x^\infty E_1(s) ds = e^{-x} - x E_1(x).
\]
In the limit, $E_1(0) = \infty$ and $E_2(0) = 1$. Furthermore,
\[
 \int_0^\infty u \, e^u \, E_1(u)^2 \, du =  2- \frac{\pi^2}{6}.
\]
See for example \cite{GellNg69} for further details.
We shall also need the equation
\[
\int_0^c e^{a x} \, E_1(a x) \, dx = 
\frac{ \gamma + \log( a c) + e^{a c} \, E_1( a c) }{a}
\]
where $a$ and $c$ are strictly positive constants, and $\gamma \approx 0.577$
is the Euler-Mascheroni constant.

\

\begin{proof}
By Theorem~\ref{t:mean1D-stat}, asymptotically 
$\EE \left[ \widehat{ \lambda(0) } \right] = \lambda$. For the variance, 
by Theorem~\ref{t:DTFE-var}, we need to evaluate two further integrals. Now,
argueing as in the proof of Theorem~\ref{t:mean1D}, 
\[
\int_A \EE \left[ g^2(x_0 \mid x, \Phi \cup \{ x \})  \right] \lambda(x) \, dx 
 = 
\int_{-w}^{x_0} \int_{x_0}^w \frac{\Lambda(t,s) \, \lambda(s) \, 
\lambda(t)}{(s-t)^2} \, e^{-\Lambda(t,s)} \,  dt \, ds + 
\]
\begin{equation}
\label{e:var1Dhelp}
\frac{\Lambda(-w,w) \, e^{-\Lambda(-w,w)}}{4 w^2} + 
\int_{x_0}^w \frac{ \Lambda(-w,s) \, \lambda(s)}{(w+s)^2} \, 
e^{ -\Lambda(-w,s) } \, ds 
+ \int_{-w}^{x_0} \frac{\Lambda(t,w) \, \lambda(t)}{(w-t)^2} \, 
e^{ -\Lambda(t,w) } \, dt.  
\end{equation}

Since the intensity function is constant and we took $x_0 = 0$, 
(\ref{e:var1Dhelp}) reduces to
\[
\frac{ \lambda e^{-2\lambda w}}{ 2w } +
\lambda e^{-\lambda w } \int_0^w \frac{\lambda e^{ -\lambda s }}{w+s} \, ds +
\lambda e^{-\lambda w } \int_{-w}^0 \frac{\lambda e^{ \lambda t }}{w-t} \, dt  + 
\int_{-w}^0 \int_0^w \frac{\lambda^3 e^{\lambda t} e^{-\lambda s}}{s-t} \,
dt \, ds.
\]
Clearly, the first term above converges to $0$ as $w\to \infty$.
Due to symmetry, the two middle terms are equal. Note that
\[
2 \lambda \int_0^w \frac{\lambda e^{ -\lambda (s+w) }}{s+w} \, ds
=
2 \lambda^2 \int_{\lambda w}^{2\lambda w} \frac{e^{-u}}{u} \, du
=
2 \lambda^2 \left[ E_1(\lambda w) - E_1(2\lambda w) \right],
\]
which converges to zero as $w \to \infty$.
Moreover,
\[
\lambda^3 \int_{-\infty}^0 \int_0^\infty 
  \frac{e^{\lambda t} e^{-\lambda s}}{s-t} \, dt \, ds  = 
\lambda^3 \int_{-\infty}^0 E_1(-\lambda t) \, dt =
\lambda^2 E_2(0) = \lambda^2.
\]

To calculate the double integral in Theorem~\ref{t:DTFE-var}, let $x \neq y$
be points of $(-w,w)$, fix $x_0 \not \in \{ x, y, -w, w \}$, and let $\varphi$ 
be a realisation of $\Phi$, which we augment by $-w$ and $w$ in order to 
obtain bounded Delaunay cells. Since almost surely none of $x$, $y$ or $x_0$ 
lie in $\Phi$, assume $x_0, x, y \not \in \varphi$, and consider 
$g(x_0 \mid x, \varphi \cup \{ x, y \})$ as defined in (\ref{e:DTFEkernel}). 
Note that $x_0$ belongs to a single Delaunay cell interior. If $x$ and $y$
are not both endpoints of this cell, $g(x_0 \mid x, \varphi \cup \{ x, y \}) 
\, g(x_0 \mid y, \varphi \cup \{ x, y \}) = 0$. Otherwise, without loss of 
generality, $x < x_0 < y$, and $g(x_0 \mid x, \varphi \cup \{ x, y \}) = 1 /
( y - \varphi^-(x_0) )$ and $g(x_0 \mid y, \varphi \cup \{ x, y \}) = 1 / 
( \varphi^+(x_0) - x )$. 

Thus, for $x < x_0$ and $y > x_0$, let $F^-$ and $F^+$ be the cumulative 
distribution functions of $\Phi^-(x_0)$ and $\Phi^+(x_0)$. By
Lemma~\ref{l:lifetime},
\[
\EE \left[ g(x_0 \mid x, \Phi \cup \{ x, y \} ) \,
                         g(x_0 \mid y, \Phi \cup \{ x, y \} ) \right] 
=
\int_{-w}^x \int_y^w \frac{dF^-(t) \, dF^+(s)}{ (y-t) \, (s-x) }
= 
\]
\[
\int_y^w \frac{ \lambda(s) }{(w+y) \, ( s - x) } \, 
e^{ -\Lambda(-w,s)}  \, ds +  
\int_{-w}^x \frac{ \lambda(t) }{ (y - t) \, (w-x) } \, 
e^{ -\Lambda(t,w) } \, dt 
\]
\[
+ \frac{e^{ -\Lambda(-w,w)}}{(w+y) \, (w-x) } + 
\int_{-w}^x \int_y^w \frac{ \lambda(s)\, \lambda(t) }{( y - t ) \,
(s-x)} \, e^{ -\Lambda(t,s) } \, dt \, ds.
\]
By symmetry, 
\[
\int_A \int_A
\EE \left[ g(x_0 \mid x, \Phi \cup \{ x, y \} ) \,
                         g(x_0 \mid y, \Phi \cup \{ x, y \} ) \right] 
\lambda(x) \, \lambda(y) \, dx \, dy = 
\]
\begin{eqnarray}
\nonumber
& &
2 e^{-\Lambda(-w,w)} \int_{-w}^{x_0} \frac{\lambda(x)}{w-x} \, dx
\int_{x_0}^w \frac{\lambda(y)}{w+y} \, dy 
\\
& + & 
\nonumber
2 \int_{x_0}^w \lambda(s) e^{-\Lambda(-w,s)} \left[ 
\int_{-w}^{x_0}  \frac{\lambda(x)}{s-x} \, dx
\int_{x_0}^s \frac{\lambda(y)}{w+y} \, dy
\right] ds \\
& + &
\nonumber
2 \int_{-w}^{x_0} \lambda(t) e^{-\Lambda(t,w)} \left[ 
\int_{t}^{x_0} \frac{\lambda(x)}{w-x} \, dx
 \int_{x_0}^w \frac{\lambda(y)}{y-t} \, dy
\right] dt \\
& + &
2 \int_{-w}^{x_0} \int_{x_0}^{w}
\lambda(t) \, \lambda(s) \, e^{-\Lambda(t,s)} \left[
 \int_t^{x_0} \frac{\lambda(x)}{s-x} \, dx 
\int_{x_0}^s \frac{\lambda(y)}{y-t} \, dy \right] dt \, ds.
\label{e:gg}
\end{eqnarray}
For $x_0 \in \{ -w, w \}$, formula (\ref{e:gg}) holds true under 
the convention that integrals over intervals of zero length vanish,
as in this case $x_0$ cannot belong to any Delaunay cell with 
endpoints $x < x_0 < y$.

Next, we plug in $x_0 = 0$ and $\lambda(\cdot) \equiv \lambda$, and
consider each integral in (\ref{e:gg}) in turn. 
The main term is the four fold integral 
\[
 \int_{-w}^0 \int_0^w
\int_{-w}^x \int_y^w \frac{ 2 \lambda^4 \, e^{ \lambda t } e^{ - \lambda s }}
{( y - t ) \, (s-x)} \,  dx \, dy \, dt \, ds.
\]
Its limit as $w\to \infty$ is
\[
2 \, \lambda^4 \int_{-\infty}^0 \int_0^\infty e^{\lambda(y-x)} 
\left[
\int_{-\infty}^x \frac{ e^{-\lambda(y-t)} } { y-t } \, dt 
\int_{y}^\infty \frac{ e^{-\lambda(s-x)} } { s-x } \, ds 
\right] dx \, dy  = 
\]  
\[
2 \, \lambda^4 \int_{-\infty}^0 \int_0^\infty e^{\lambda(y-x)} 
E_1( \lambda(y-x) )^2 \, dx \, dy = 
 2 \, \lambda^3 \int_{-\infty}^0 \int_{-\lambda x}^\infty e^{u} E_1(u)^2 dx \, du 
 =   
\]
\[
2 \, \lambda^2 \int_0^\infty \int_y^\infty e^{u} E_1(u)^2 dy \, du =
2 \, \lambda^2 \int_0^\infty u \, e^u \, E_1(u)^2 \, du =
 2 \, \lambda^2 (2-\pi^2/6 ) ,
\]
upon a change of integration order.

The first term in (\ref{e:gg}) reduces to 
\(
2 e^{-2\lambda w} ( \lambda \log 2)^2
\)
for a homogeneous Poisson process, which tends to zero as $w\to \infty$.

It remains to consider the sum of the two three fold integrals in (\ref{e:gg})
\[
 \int_{-w}^0 \int_0^w \int_y^w
\frac{4\lambda^3 \, e^{-\lambda (s+w)} }{(s-x) \, (y+w)} \, dx \, dy \, ds
\]
which can be written as
\[
4 \, \lambda^3 \int_0^w \int_0^w  \left( \int_0^s \frac{dy}{y+w} \right)
\frac{e^{-\lambda (s+w)} }{s+x } \, dx \, ds \leq
4 \, \lambda^3 \, \log 2 \int_0^w e^{-\lambda w + \lambda x} \left(
\int_0^w \frac{e^{-\lambda (s+x)} }{s+x } \, ds \right) dx 
=
\]
\[
4 \, \lambda^3 \, \log 2 \int_0^w e^{-\lambda w + \lambda x} 
\left[ E_1(\lambda x) - E_1(\lambda x + \lambda w) \right] dx
= 4 \, \lambda^2 \, h( \lambda, w ) \log 2,
\]
where 
\begin{eqnarray*}
h(\lambda, w) & = &
e^{-\lambda w} \int_0^{\lambda w} e^u 
\left[ E_1(u) - E_1(u + \lambda w) \right] du \\
& = &
\left( e^{-\lambda w} + e^{-2 \lambda w} \right) 
\int_0^{\lambda w} e^u \, E_1(u) \, du 
- e^{-2 \lambda w} \int_0^{2 \lambda w} e^u \, E_1(u) \, du \\
& = &
e^{-\lambda w} \gamma + ( e^{-\lambda w} + e^{-2\lambda w} ) \log(\lambda w) 
-  e^{-2\lambda w} \log( 2 \lambda w) \\
& + & E_1(\lambda w) ( 1 + e^{-\lambda w} )
- E_1(2 \lambda w)
\end{eqnarray*}
tends to zero as $w\to \infty$. The proof is finished upon
collection of all terms.
\end{proof}

As a corollary, the proof gives an expression for the second moment
of the Delaunay tessellation field estimator of the intensity function
for Poisson processes with not necessarily constant locally finite 
intensity function on intervals of the form $[-w, w]$ by combining 
(\ref{e:var1Dhelp})--(\ref{e:gg}). A slightly simpler proof can be 
obtained by an appeal to Theorem~\ref{t:asymptotic}, but such a proof 
cannot be generalised to non-homogeneous Poisson processes.

\section{Discussion}

In this paper, we analysed Schaap and Van de Weygaert's Delaunay
tessellation field estimator \cite{SchaWeyg00,Scha07} for the 
intensity function of a point process. We expressed its mean and 
variance in terms of the first and second order factorial moment 
measures of the underlying point process, and placed the estimator
in the context of adaptive kernel estimation. We then focussed on
Poisson point processes, and showed that for stationary Poisson
processes, the DTFE is asymptotically unbiased with a variance 
that is proportional to the squared intensity. The proportionality
constant depends on the dimension. For $d=1$, explicit calculation
is possible. For $d=2$, we used the DELDIR package \cite{deldir}
to obtain $C(\lambda, 2) \approx 0.8 \lambda^2$ and 
$C^\prime(\lambda, 2) \approx 0.6 \lambda^2$, see Lemma~\ref{e:C}
and \ref{e:C'}. Note that in the plane it is possible to write 
mean and variance as repeated integrals in the spirit of Calka 
\cite{Calk03}, but explicit evaluation seems difficult. Simulations 
for the case $d=3$ of most interest to cosmologists can be found in 
Schaap's Ph.D.\ thesis \cite{Scha07}.

\section*{Acknowledgement}
The author is grateful to Dr.\ N.M. Temme for access to \cite{GellNg69}.

\end{document}